\newif\ifsocgproc
\socgprocfalse
\newif\iffullver
\fullvertrue

\ifsocgproc 	%SoCG version
\else	%Full version
	\documentclass[11pt]{article}

	\usepackage[margin=1.25in]{geometry}
	\usepackage{mathpazo}
 	\usepackage{calrsfs}
\fi

\usepackage{hyperref}
\hypersetup{
    colorlinks=true,
    linkcolor=blue
    }

\usepackage{latexsym,color,graphicx,amsmath,amssymb}
\usepackage{microtype,xspace,url}
\usepackage{amsmath,amsthm}
\usepackage{mathtools}
\usepackage{cite}
\usepackage{comment}
\usepackage{enumerate}

\long\def\socgcomm#1{#1}

\def\reals{{\mathbb R}}
\def\sphere{{\mathbb S}}
\def\Fspace{{\mathbb F}}
\def\Kspace{{\mathbb K}}
\def\eps{{\varepsilon}}
\def\bd{{\partial}}
\def\dist{{\rm dist}}
\def\A{\mathcal{A}}

\def\F{\mathcal{F}}

\def\G{\mathcal{G}}
\def\K{\mathcal{K}}

\def\M{\mathcal{M}}
\def\R{\mathcal{R}}

\def\U{\mathcal{U}}

\def\cell{\tau}

\def\XX{\mathsf{X}}
\def\NN{\mathsf{N}}

\def\ceil#1{\lceil #1\rceil}

\def\VeD#1{{#1}^{\mbox{\raisebox{4pt}{\tiny$\|$}}}}
\def\Ak#1{\A_{{\le}#1}}

\def\etal{\textsl{et~al.}}

\newcommand{\old}[1]{{{}}}

\ifsocgproc
	\authorrunning{P. K. Agarwal, M. J. Katz and M. Sharir}
	\Copyright{Pankaj K. Agarwal, Matthew J. Katz and Micha Sharir}


	\author{Pankaj K. Agarwal}
    		{Department of Computer Science, Duke University, Durham, NC 27708, USA}
    		{pankaj@cs.duke.edu}
    			{https://orcid.org/0000-0002-9439-181X}
    		{Partially supported by NSF grants CCF-20-07556, and CCF-22-23870.}
		\author{Matthew J. Katz}
    		{The Stein Faculty of Computer and Information Science, Ben-Gurion University of the Negev, Beer Sheva, Israel}
    		{matya@bgu.ac.il}
    		{https://orcid.org/0000-0002-0672-729X}
    		{Partially supported by Israel Science Foundation Grant 495/23.}
	\author{Micha Sharir}
    		{School of Computer Science, Tel Aviv University, Tel Aviv, Israel}
    		{michas@tauex.tau.ac.il}
    		{https://orcid.org/0000-0002-2541-3763}
    		{Partially supported by Israel Science Foundation Grant 495/23.}


	\title{Dynamic Nearest-Neighbor Searching Under General Metrics in $\reals^3$ and Its Applications}

	\titlerunning{Dynamic Nearest-Neighbor Searching Under General Metrics in $\reals^3$}

	\relatedversiondetails{Full Version}{https://arxiv.org/abs/0000.00000}

	\EventEditors{Hee-Kap Ahn, Michael Hoffmann, and Amir Nayyeri}
	\EventNoEds{3}
	\EventLongTitle{42nd International Symposium on Computational Geometry
		(SoCG 2026)}
	\EventShortTitle{SoCG 2026}
	\EventAcronym{SoCG}
	\EventYear{2026}
	\EventDate{June 2--5, 2026}
	\EventLocation{New Brunswick, NJ, USA}
	\EventLogo{socg-logo.pdf}
	\SeriesVolume{367}
	\ArticleNo{XX}     % <-- This will be filled in by the typesetters
\else
\newtheorem{theorem}{Theorem}[section]
\newtheorem{lemma}[theorem]{Lemma}

\newtheorem{corollary}[theorem]{Corollary}
 
\theoremstyle{definition}
\newtheorem{remark}{Remark}

\makeatletter
\long\def\@makecaption#1#2{
   \vskip 10pt
   \setbox\@tempboxa\hbox{{\footnotesize {\bf #1.} #2}}
   \ifdim \wd\@tempboxa >\hsize         % IF longer than one line:
       {\footnotesize {\bf #1.} #2\par}% THEN set as ordinary paragraph.
     \else                              %   ELSE  center.
       \hbox to\hsize{\hfil\box\@tempboxa\hfil}
   \fi}
\makeatother

\title{Dynamic Nearest-Neighbor Searching Under General Metrics in $\reals^3$ and Its Applications\thanks{%
  Work by Pankaj Agarwal has been partially supported by NSF grants
  CCF-22-23870  and IIS-24-02823, and by US-Israel Binational Science
  Foundation Grant 2022/131. 
  Work by Matthew Katz has been partially supported by Israel Science Foundation Grant 495/23.
  Work by Micha Sharir has been partially supported by Israel Science Foundation Grant 495/23.
  }
  }

\author{
    Pankaj K. Agarwal\thanks{
    Department of Computer Science, Duke University, Durham, NC 27708, USA;
    {\sf pankaj@cs.duke.edu}}
\and
    Matthew J. Katz\thanks{
    The Stein Faculty of Computer and Information Science, Ben-Gurion University of the Negev, Beer Sheva, Israel;
    {\sf matya@bgu.ac.il}}
\and
    Micha Sharir\thanks{
    School of Computer Science, Tel Aviv University, Tel Aviv, Israel;
    {\sf michas@tauex.tau.ac.il}}
}
\fi

\begin{document}

\ifsocgproc
\nolinenumbers
\fi
\maketitle

\begin{abstract}
	Let $K$ be a compact, centrally-symmetric, strictly-convex region in $\reals^3$, which is 
a semi-algebraic set of constant complexity, i.e. the unit ball of a corresponding metric,
denoted as $\|\cdot\|_K$. Let $\K$ be a set of $n$ homothetic copies of $K$.
This paper contains two main sets of results:
	\begin{enumerate}[(i)] 
	\item   For a storage parameter $s\in[n,n^3]$, $\K$ can be preprocessed in 
		$O^*(s)$ expected time into a data structure of size $O^*(s)$, so 
		that for a query homothet $K_0$ of $K$,
		an intersection-detection query (determine whether $K_0$ intersects any member of $\K$, 
		and if so, report such a member) or a nearest-neighbor query (return
                the member of $\K$ whose $\|\cdot\|_K$-distance from $K_0$ is smallest) can be 
		answered in $O^*(n/s^{1/3})$ time; all $k$ homothets of $\K$ intersecting $K_0$ can be 
		reported in additional $O(k)$ time. In addition, the data structure supports
		insertions/deletions in $O^*(s/n)$ amortized expected time per operation. Here the 
		$O^*(\cdot)$ notation hides factors of the form $n^\eps$, where $\eps>0$ 
		is an arbitrarily small constant, and the constant of proportionality depends on $\eps$.
	\item Let $\G(\K)$ denote the intersection graph of $\K$.  Using the above data 
		structure, breadth-first or depth-first search  on $\G(\K)$ can be 
		performed in  $O^*(n^{3/2})$ expected time.  Combining this result with the so-called 
		shrink-and-bifurcate technique,
			%~\cite{BFKKS,CH,KKSS}, 
                the reverse-shortest-path problem in a suitably defined proximity graph of $\K$ can 
		be solved in $O^*(n^{62/39})$ expected time.
		Dijkstra's shortest-path algorithm, as well as Prim's MST algorithm, on a 
	        $\|\cdot\|_K$-proximity graph on $n$ points in $\reals^3$, 
		with edges weighted by $\|\cdot\|_K$, can also be performed in $O^*(n^{3/2})$ time.
	\end{enumerate}
\end{abstract}

%-------------------------------------
\section{Introduction}
\label{sec:intro}

\ifsocgproc 
\subparagraph*{Problem statement.}
\else
\paragraph{Problem statement.}
\fi
Let $K$ be a compact, centrally-symmetric, strictly-convex region in $\reals^3$, which 
is a semi-algebraic set of constant complexity.\footnote{%
  Roughly speaking, a semi-algebraic set in $\reals^d$ is the set of points in $\reals^d$ that 
  satisfy a Boolean formula over a set of polynomial inequalities; the complexity of a semi-algebraic 
  set is the number of polynomials defining the set and their maximum degree.
  See \cite{BPR} for formal definitions of a semi-algebraic set and its complexity.}
$K$ is the unit ball of a norm
$\|\cdot\|_K$, and we denote its metric as $\dist_K$. 

A \emph{homothet} (or \emph{homothetic copy}) of $K$ is a scaled and translated copy of $K$, expressed as
$K(c,\rho) := \rho K + c$, for $c\in\reals^3$ and $\rho \ge 0$; $c$ is the \emph{center}
of $K(c,\rho)$ and $\rho$ is its \emph{size}. We represent $K(c,\rho)$ as the point 
$(c,\rho)\in\reals^3\times\reals_{\ge 0}$. Let $\Kspace$ be the set of all homothets of 
$K$, represented as the halfspace $x_4\ge 0$ of $\reals^4$.  
The \emph{$K$-distance} $\dist_K(u,v) = \|u-v\|_K$ between points $u,v\in\reals^3$
is the minimum value of $\rho$ for which 
$v\in K(u,\rho)$ (or, equivalently, $u\in K(v,\rho)$). Since $K$ is strictly convex, the 
triangle inequality $\dist_K(u,v) \le \dist_K(u,w) + \dist_K(w,v)$ is sharp unless $w$ lies 
on the segment $uv$.
The $K$-distance of a point $x\in\reals^3$ from a homothet $K(c,\rho)$ is defined as 
$f_{c,\rho}(x) = \dist_K(c,x)-\rho$, 
and the $K$ distance of another homothet $K(c',\rho')$ from $K(c,\rho)$ is $f_{c,\rho}(c')-\rho' = \dist_K(c,c')-\rho-\rho'$. 
Since $K$ is a constant-complexity semi-algebraic set, $f_{c,\rho}$ is 
a semi-algebraic function of constant complexity. 
The homothets $K(c,\rho)$ and $K(c',\rho')$ intersect if the distance between them is non-positive, 
i.e., $\dist_K(c,c') \le \rho+\rho'$, or equivalently, the point $(c',\rho')$ lies 
above the graph of $f_{c,\rho}$ in $\reals^4$, or vice-versa.

Let $\K=\{K_1, \ldots, K_n\} \subset \Kspace$ be a set of $n$ homothets of $K$, where 
$K_i=K(c_i,\rho_i)$. In this paper we study two classes of proximity problems related to $\K$.

\medskip
\begin{description}
	\item[Proximity queries.] We wish to preprocess $\K$ into a data structure so that 
                various proximity queries for a query homothet $K_0$, such as intersection-detection 
		(determine whether $K_0$ intersects any homothet of $\K$; if the answer is yes, return 
                one such intersecting homothet of $\K$), 
		intersection reporting (report all homothets of $\K$ that intersect $K_0$),
		and nearest-neighbor queries (report the homothet of $\K$ at the 
		minimum $K$-distance from $K_0$),
		can be performed efficiently. In addition, the data structure should 
		support insertions and deletions of homothets of $K$.
	\item[Searching in the intersection graph and related problems.] Let $\G(\K)$ denote the intersection graph 
		of $\K$, whose vertices are the homothets of $\K$ and $(K_i,K_j)$ is an edge if 
		$K_i\cap K_j\ne\emptyset$. Similarly, for a threshold parameter $r_0$, the
                \emph{$r_0$-proximity graph} $\G_{r_0}(\K)$ of $\K$ consists of those edges $(K_i,K_j)$
                for which $\dist_K(K_i,K_j) \le r_0$.
% \footnote{Given a set $\Gamma$ of geometric objects 
		% in $\reals^d$, the \emph{intersection graph} $\G(\Gamma)$ of $\Gamma$ has 
		% $\Gamma$ as vertices, and there is an edge $(\gamma_i, \gamma_j)$ in the 
		% graph if $\gamma_i\cap\gamma_j\ne \emptyset$. Given a distance function 
		% $\dist: \Gamma \rightarrow \reals$ and a threshold $r_0$, the 
		% \emph{proximity graph} on $\Gamma$ has an edge $(\gamma_i,\gamma_j)$ if 
		% $\dist(\gamma_i,\gamma_j) \le r_0$.}
		We study efficient implementation of graph-search problems on $\G(\K)$,
                like BFS or DFS.
		We also study the \emph{reverse shortest path} (RSP) problem on $\K$:
		given two elements $K_i,K_j \in \K$  and a parameter $1 \le k < n$, compute
		the smallest value $r^*$ for which $\G_{r^*}(\K)$
% , where $\K_r = \{K(c_i,r_i+r) \mid 1 \le i \le n\}$, 
                contains a path from $K_i$ to $K_j$ with at most $k$ edges. 
\end{description}

%Let $r > 0$ be a parameter, and let $\G_r(\K)$ be the intersection graph of $\K_r$, where $\K_r = \{ K(c,\rho+r) \mid K(c,\rho) \in \K \}$; that is, $\K_r$ is 
%the set of homothetic copies of $K$ that is obtained from $\K$ by expanding each of the elements of $\K$ by $r$. In the RSP problem, we are given 
		%The decision problem associated with the RSP problem (i.e., given $r$, decide whether $\G_r(\K)$ contains such a path) is easily solved by performing a BFS in $\G_r(\K)$ from $K(c_s,\rho_s+r)$, and we employ the \emph{shrink-and-bifurcate} technique~(\cite{BFKKS,CH,KKSS}) to obtain an efficient solution to the RSP problem.      

%-------------------------------------
\ifsocgproc 
\subparagraph*{Related work.}
\else
\paragraph{Related work.}
\fi
There is extensive work on nearest-neighbor (NN) searching in many different fields including
computational geometry, database systems, and machine learning. We refer the readers to 
various survey papers~\cite{AE99,AI17,Cl} for a general overview of known results on this topic. 
Here we briefly mention a few results that are most closely related to the problems studied 
in this paper. Agarwal and Matou\v{s}ek~\cite{AM95} presented a dynamic NN data structure for 
a point set in $\reals^d$, under the Euclidean metric, that answers a query in $O^*(n/s^\phi)$ 
time, where $\phi=1-\tfrac{1}{\ceil{d/2}}$ and $s \in [n,n^{\ceil{d/2}}]$ is a storage parameter, 
using $O^*(s)$ space and preprocessing; the (amortized) update time for insertions/deletions of 
points is $O^*(s/n)$. The bounds were slightly improved -- $n^\eps$ factors were replaced by 
$\log^{O(1)}n$ factors -- in~\cite{Chan:le,KMRSS}.
Kaplan~\etal~\cite{KMRSS} presented a dynamic NN data structure for a set $S$ of points in 
$\reals^2$ under fairly general distance functions. Its performance depends on the 
complexity of the Voronoi diagram of $S$ under that distance function. For the Minkowski 
metric induced by a centrally-symmetric convex region, which is a semi-algebraic set of
constant complexity, their data structure can answer an NN query in $O(\log^2n)$ time, and handle 
insertion and deletion of a point in $O(\beta(n)\log^5n )$ and $O(\beta(n)\log^9n)$ 
amortized expected time, respectively, where $\beta(n)$, a variant of inverse Ackermann's function, is
an extremely slowly growing function; see also~\cite{Liu} for slightly improved bounds. 
These data structures, however, do not extend to $\reals^3$.
%which is what we achieve in this work.

The problem of NN searching under general distance functions
in $\reals^3$ is much less understood.
By reducing NN queries to ball-intersection queries~\cite{AM93} and 
using semi-algebraic range-searching data structures~\cite{AAEZ,MP,AAEKS}, an NN query 
amid the set $\K$ of $n$ homothets in $\reals^3$, as above, can be answered in $O^*(1)$ time 
using $O^*(n^4)$ space and preprocessing, or in $O^*(n^{3/4})$ time using $O(n)$ space and preprocessing. 
Using recent standard techniques (see, e.g., the Appendix in~\cite{AAEKS}),
one can obtain a space/query-time trade-off, as well as handle insertions/deletions of objects.
Using a recent result by Agarwal~\etal~\cite{AES:vd} on the vertical decomposition of the lower 
envelope of trivariate functions, an NN query amid $\K$ can be answered in $O(\log^2 n)$ time 
using only $O^*(n^3)$ space, but this data structure does not support efficient deletions 
(insertions can be handled in a standard manner, using the dynamization technique 
of Bentley and Saxe~\cite{BS80}). Furthermore, it was not known whether an NN query 
amid $\K$ can be answered in $O^*(n^{2/3})$ time using $O^*(n)$ space.

Motivated by numerous applications, there has been work on developing fast algorithms 
for intersection and proximity graphs of $n$ geometric objects. Although the intersection 
graph may have $\Theta(n^2)$ edges, the goal in this line of work is to develop algorithms, 
by exploiting the underlying geometry, that run in $O^*(n)$ time (or at least in subquadratic time).
Cabello and Jej\v{c}i\v{c}~\cite{CabelloJ15}, and subsequently Chan and Skrepetos~\cite{CS}, presented
$O(n \log n)$ implementations of BFS in unit-disk intersection graphs, which was recently 
extended to general disk-intersection graphs by de Berg and Cabello~\cite{BergC25}; 
see also ~\cite{KKSS,KMRSS,Klost23}. The algorithm in~\cite{BergC25} can also perform DFS 
and Dijkstra's algorithm in $O(n\log n)$ time.
%(Balls are a special (much easier) case of homothetic copies of a centrally-symmetric convex object $K$.) 
Katz~\etal~\cite{KSS} showed that BFS in ball-intersection graphs (for congruent or 
arbitrary Euclidean balls)
can be implemented in $O^*(n^{2\phi/(\phi+1)})$ time, where $\phi = \lfloor d/2 \rfloor + 1$, 
so in $O^*(n^{4/3})$ time for $d=3$. The problem of devising efficient implementations for BFS/DFS 
in more general graphs in the plane also has received some attention, 
see, e.g., \cite{AgarwalKKS24,AgarwalKS24,KSS24}.  

The reverse shortest path (RSP) problem for unit disks was studied by 
Wang and Zhao~\cite{WZ}, who gave an $O^*(n^{5/4})$-time solution. 
Kaplan~\etal~\cite{KKSS} improved the runtime to $O^*(n^{6/5})$ using the 
\emph{shrink-and-bifurcate} technique (see~\cite{BFKKS}), and it was recently improved to $O^*(n^{9/8})$ by
Chan and Huang~\cite{CH}. The best-known RSP algorithm for arbitrary disks runs in $O^*(n^{6/5})$ expected time.
%using Chan and Huang's technique~\cite{CH}. Katz et al.~\cite{KSS} studied the RSP problem for balls in $\reals^d$, for $d \ge 3$. 
Katz~\etal~\cite{KSS} presented RSP algorithms for balls in $\reals^3$ --- with 
$O^*(n^{29/21})$ time for congruent balls and $O^*(n^{56/39})$ for arbitrary balls.  
The RSP problem also has been studied for other geometric objects in 
$\reals^2$~\cite{AgarwalKS24,KSS,CH}.  

\ifsocgproc 
\subparagraph*{Our results.} 
Our main result is a dynamic data structure for answering intersection 
and NN queries with a homothet of $\Kspace$ amid the set $\K$ of homothets:
\else
\paragraph{Our results.} 
Our main result is a dynamic data structure for answering intersection 
and NN queries with a homothet of $\Kspace$ amid the set $\K$ of homothets, as stated in the following theorem:
\fi

%---------------------------------------------
\begin{theorem} \label{thm:NN}
Let $K$ be a compact, centrally-symmetric strictly convex semi-algebraic set
in $\reals^3$ of constant complexity.  Let $\K$ be a set of $n$ homothetic 
copies of $K$. For a storage parameter $s \in [n,n^3]$, $\K$ can be maintained in a dynamic data 
	structure with $O^*(s)$ storage, that can answer an intersection-detection or NN-query (under the $K$-distance) 
	with a homothet $K_0$ of $K$ 
	in $O^*(n/s^{1/3})$ time and that supports insertions/deletions in $O^*(s/n)$ 
	amortized expected time. It can report all $k$ objects of $\K$ intersecting $K_0$ in additional $O(k)$~time.
\end{theorem}
%---------------------------------------------
Plugging Theorem~\ref{thm:NN} into Eppstein's technique~\cite{Epp}, or its enhancement by Chan~\cite{Chan:bcp},
for maintaining bichromatic closest pairs, using $s=n^{3/2}$, we obtain:
%--------------------------------
\begin{corollary}
	\label{cor:bcp}
Let $K$ be a compact, centrally-symmetric strictly-convex semi-algebraic set
in $\reals^3$ of constant complexity. Let $\K, \K'$ be two sets of homothets of $K$
of combined size $n$. The $K$-closest pair between $\K$ and $\K'$,
under insertions/deletions of homothets (in $\K$ and $\K'$), can be 
maintained in $O^*(n^{1/2})$ amortized expected time per update.
\end{corollary}
%--------------------------------

Theorem~\ref{thm:NN} is obtained by presenting two data structures. The first 
one (see Section~\ref{subsec:log})
uses $O^*(n^3)$ space, answers a query in $O^*(1)$ time, and handles an insertion/deletion in $O^*(n^2)$ amortized expected time. 
(The data structure in~\cite{AES:vd} also obtains a similar space/query-time bound but it cannot handle deletions efficiently.)
This data structure is based on constructing \emph{vertical shallow cuttings}, a notion originally introduced in~\cite{CT16,Chan:le},  of the graphs of the distance functions of $\K$. Roughly speaking,
for a parameter $t\ge 1$, a \emph{vertical $(n/t)$-shallow 
cutting} of $\K$ is a collection of pairwise openly disjoint semi-unbounded \emph{pseudo-prisms} 
that covers the region lying below the $k$-level of 
the arrangement of $\K$, where
each prism, consisting of all points that lie 
vertically below a constant-complexity semi-algebraic 
$3$-dimensional region, intersects the graphs of only $O(n/t)$ functions of $\K$.
The random-sampling based technique used in~\cite{Chan:le,KMRSS} to construct a vertical shallow cutting 
does not immediately extend to our setting because one 
needs to decompose the $q$-level, for some parameter $q>0$, in the 
4D arrangement of the distance functions of $\K$ 
into $O^*(n^3(q+1))$ constant-complexity cells. Although a recent 
result~\cite{AES:vd} constructs such a decomposition of the $0$-level (with $O^*(n^3)$ cells), it does not 
extend to larger values of $q$.
This problem remains elusive for arrangements of arbitrary trivariate semi-algebraic functions of constant complexity, 
but we develop a desired decomposition technique for our setting
(in Section~\ref{sec:vd}) by exploiting the 
structural geometric properties of the distance functions of $\K$. This result is one 
of the main technical contributions of the paper. Using our result on the decomposition of the $q$-level,
we construct a vertical $(n/t)$-shallow 
cutting of $\K$ of size $O^*(t^3)$, for any $t$ (Section~\ref{sec:shallow}). 
\iffullver
We build a tree data structure using vertical shallow cuttings
and adapt the technique in~\cite{AM95} to handle efficient deletions of 
homothets from $\K$; as noted, insertions are handled using the standard Bentley-Saxe technique~\cite{BS80}.
\fi

The second data structure is a linear-size partition tree, with $O^*(n^{2/3})$ query time,
constructed on the point set $\K^* = \{ (c_i, \rho_i) \mid 1 \le i \le n\} \subset\reals^4$,
that answers intersection-detection queries on $\K$ with a homothet of 
$\Kspace$ as a query. The main 
technical challenge we face here is the construction of a so-called \emph{test set} of $\K$ 
(see \cite{ShSh} and below) of size $r^{O(1)}$, which represents well the distance functions of all 
\emph{$(n/r)$-shallow} homothets in $\Kspace$, i.e., the homothets that intersect 
at most $n/r$ elements of $\K$. By adapting the technique in~\cite{ShSh} and again exploiting 
the structural geometric properties of the distance functions, we show 
\iffullver
(in Section~\ref{subsec:linear}) 
\fi
that a test set, of size $O^*(r^4)$, with the desired properties can be constructed efficiently.

\iffullver
Using Theorem~\ref{thm:NN} and Corollary~\ref{cor:bcp}, we obtain our second set of results 
related to searching in the intersection and proximity graphs of $\K$.
\else
Using Theorem~\ref{thm:NN} and Corollary~\ref{cor:bcp}, we obtain our second set of results:
\fi
%---------------------------------------------
\begin{theorem} \label{thm:bfs}
(a) Given a set $\K$ of $n$ homothets of a constant-complexity, strictly convex, centrally-symmetric semi-algebraic set $K$,
BFS or DFS in the intersection graph of $\K$ can be performed in $O^*(n^{3/2})$ expected time. 

(b) The single-source shortest-path tree as well as the minimum spanning forest of the
$r_0$-proximity graph of a set of $n$ points, for any threshold parameter $r_0$, 
	weighted by the pairwise $K$-distances, can be computed in $O^*(n^{3/2})$ time.
\end{theorem}
%---------------------------------------------
%In particular, Theorem~\ref{thm:bfs}(b) applies to the $r_0$-proximity graph on any set of $n$ points in the plane, under the $K$-distance.

Combining Theorem~\ref{thm:bfs}(a) with the shrink-and-bifurcate 
technique~\cite{BFKKS,KKSS} yields:

%---------------------------------------------
\begin{theorem} \label{thm:rsp}
Given a set $\K$ of $n$ homothets of a constant-complexity, strictly convex, centrally-symmetric semi-algebraic 
set $K$, two designated elements $K$, $K'\in \K$, and an integer $k < n$, the RSP problem on $\K$, of finding
the smallest $r^*$ for which $G_{r^*}(\K)$ contains a path between $K$ and $K'$ of length at most $k$,
can be solved in $O^*(n^{62/39})$ expected time.
\end{theorem}
%---------------------------------------------
\iffullver
\else
Because of lack of space, the proofs of Theorems~\ref{thm:bfs} and~\ref{thm:rsp} and of some of the lemmas are 
omitted from this version.
\fi

%------------------------------------------------------------
\ifsocgproc 
\section{Decomposing the $k$-Level of an Arrangement of Distance Functions} \label{sec:vd}
\else
\section{Decomposition of the $k$-Level of an Arrangement of Distance Functions} \label{sec:vd}
\fi

\ifsocgproc 
\subparagraph*{Preliminaries.}
\else
\paragraph{Preliminaries.}
\fi
Let $\K = \{K_1, \ldots, K_n\}$, where each $K_i$ 
is a homothetic copy of $K$ represented by a point $K_i^* = (c_i,\rho_i)$ in $\reals^4$.
%For technical reasons, as well as simplicity of the exposition, we assume that the homothets in $\K$ are \emph{non-nested}, i.e., $K_i \not\subseteq K_j$ for any pair $i\ne j$.
% Later we will discuss how to handle nested homothets.
% \pankaj{Make sure that we do this.} \micha{So far we don't.}
%Handling nested copies require additional work, which we omit in this version.
For $1 \le i \le n$, set $f_i = f_{c_i,\rho_i}$, and let 
$\F := \F(\K) = \{ f_i \mid 1 \le i \le n\}$.
(As noted, distances from points $x$ that lie inside some copy $K_i$, still measured by $f_i(x)$, are nonpositive.)
We will not distinguish between a function of $\F$ and its graph.
For a point $x\in\reals^3$, $K_i$ is the \emph{nearest neighbor} of $x$ in $\K$ 
(under the $K$-distance $\dist_K$), i.e., $i = \arg\min_{1 \le j \le n} f_j(x)$, if
$f_i$ appears on the lower envelope of $\F$ at $x$. 

The \emph{level} of a point $p\in\reals^4$ in $\A(\F)$ is the number of functions in $\F$ whose graphs lie below $p$.
For a parameter $k < n$, the \emph{$k$-level} of $\A(\F)$, denoted $\A_k(\F)$, is the 
(closure of the) locus of all points on $\bigcup \F$ whose level is $k$.
We define the \emph{$(\le k)$-level} of $\A(\F)$, denoted $\A_{\le k}(\F)$, to be 
the (closure of the) set of all points in $\reals^4$ of level at most $k$, i.e., 
the set of points that lie on or below $\A_k(\F)$ (not necessarily on $\cup\F$).
The projection of $\A_k(\F)$ onto $\reals^3$, 
denoted by $\M_k = \M_k (\K)$, is a subdivision of $\reals^3$. 
If a cell $\cell^\downarrow$ of $\M_k$ is the projection of a cell $\cell$ of 
$\A_k(\F)$ that lies on the graph of $f_i$, 
then $K_i$ is the $k$-th nearest neighbor in $\K$ (under the $K$-distance function)
of all points in $\cell^\downarrow$. Note that $A_0(\F)$ is the lower envelope of $\F$, 
and $\M_0(\K)$ is the Voronoi diagram of $\K$ (under the $K$-distance).

\ifsocgproc 
\subparagraph*{The decomposition.}
\else
\paragraph{The decomposition.}
\fi
The main technical result of this section is that $\A_k(\F)$ can be decomposed into 
$O^*(n^3(k+1))$ \emph{elementary cells}, namely, constant-complexity semi-algebraic regions, 
each homeomorphic to a ball; see below for a more precise definition.
%We begin with two auxiliarly lemmas:
%the following simple observation. 
%The proofs of the following sequence of lemmas are delegated to Appendix~\ref{sec:proof}.

\iffullver
%--------------------------------
\begin{lemma}
\label{lem:nested}
$K_i$ is contained in $K_j$ if and only if $\dist_K(c_j,c_i) \le \rho_j -\rho_i$.
\end{lemma}
%--------------------------------
\begin{proof}
Suppose that $\dist_K(c_j,c_i) > \rho_j -\rho_i$. Let $x$ be the intersection point of 
$\bd K_i$ with the ray emanating from $c_i$ in direction $c_i-c_j$, i.e., lying on the 
line passing through $c_i$ and $c_j$ but not containing $c_j$. Because of the collinearity 
of $c_i$, $c_j$, and $x$,
\[ 
\dist_K(c_j, x) = \dist_K(c_j,c_i) + \dist_K(c_i,x) = \dist_K(c_j,c_i)+\rho_i > \rho_j.
\]
Hence, $x\not\in K_j$ and thus $K_i \not\subseteq K_j$. Conversely, assume 
that $\dist_K(c_j,c_i) \le \rho_j -\rho_i$.
Let $x\in K_i$ be an arbitrary point. That is, $\dist_K(x,c_i) \le \rho_i$.
By the triangle inequality we have
\begin{align*} 
	\dist_K(x,c_j) - \rho_j & \le \dist_K(x,c_i) + \dist_K(c_i,c_j) - \rho_j \\
		& \le \dist_K(x,c_i) + (\rho_j - \rho_i) - \rho_j \\
		& = \dist_K(x,c_i) - \rho_i \le 0,
\end{align*}
i.e., $x\in K(c_j,\rho_j)$. Since this holds for any point $x\in K_i$, we conclude that
$K(c_i,\rho_i) \subseteq K(c_j,\rho_j)$.
\end{proof}

%---------------------------------------
\begin{lemma} \label{lem:c0in}
	Let $K_i$ and $K_j$ be a pair of non-nested homothets in $\K$ (i.e., neither of them lies inside the other). Then $f_i(c_i) < f_j(c_i)$.
%Assume that the homothets in $\K$ are pairwise non-nested. Then, for any $K_i = K(c_i,\rho_i) \in \K$, 
%$K_i$ is the nearest neighbor of $c_i \in \reals^3$ 
%(under the distance function $\dist_K$), i.e., $f_i$ attains the lower envelope of $\F$ at $c_i$.
\end{lemma}
%---------------------------------------
\begin{proof}
	Assume to the contrary that $f_j(c_i) \le f_i(c_i)$.
	%$K_j= (c_j,\rho_j)$, for some $j \ne i$, is the nearest neighbor of $c_i$. 
	That is,
\[
\dist_K(c_j,c_i) - \rho_j \le \dist_K(c_i,c_i) - \rho_i = -\rho_i , \qquad\text{or}\qquad
\dist_K(c_j,c_i) \le \rho_j -\rho_i .
\]
But then, by Lemma~\ref{lem:nested}, $K_i \subseteq K_j$, contradicting the assumption that they
are non-nested.
\end{proof}
\fi

For $1 \le i \ne j \le n$, let $B_{ij} = \{ x \in \reals^3 \mid f_i(x) = f_j(x) \}$ 
be the \emph{bisector} of $K_i$ and $K_j$ under the $K$-distance.  Note that if $K_i \subset K_j$ then,
as is easily verified, $f_j(x) < f_i(x)$ for all $x\in\reals^3$ and $B_{ij}$ is undefined.
The following lemma gives a useful property of bisectors.
%Its proof is presented in the full version.
%Its proof is presented in Appendix~\ref{sec:proof}.
%---------------------------------
\begin{lemma}
\label{lem:mono}
	Fix a homothet $K_i\in\K$. Let $\lambda$ be a ray in $\reals^3$ emanating from $c_i$. 
	For any $1 \le j \ne i \le n$ such that $K_i \not\subset K_j$, $\lambda$ intersects $B_{ij}$ in at most one 
	point, say, $\xi$. 
Furthermore, $f_i(x)\le f_j(x)$ for all $x\in c_i\xi$ (where $c_i\xi$ denotes the segment with endpoints $c_i$ and $\xi$) and $f_j(x)>f_i(x)$ for all $x \in \lambda\setminus c_i\xi$.
%(A similar property holds for rays emanating from $c_j$.)
\end{lemma}
%---------------------------------

\iffullver
\begin{figure}[htb]
  \begin{center}
	  \includegraphics[scale=0.8]{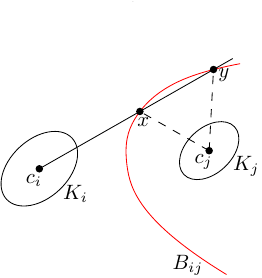} 
  \caption{Illustrating (an impossible scenario in) the proof of Lemma~\ref{lem:mono}.} 
  \label{fig:cc0}
  \end{center}
\end{figure}

\begin{proof}
Suppose to the contrary that $\lambda$ intersects $B_{ij}$ at 
two points $x\ne y$, and assume, without loss of generality, that $x$ is closer to 
$c_i$ than $y$; see Figure~\ref{fig:cc0}. We then have
\begin{align*}
\dist_K(x,c_i) - \rho_i & = \dist_K(x,c_j) - \rho_j \\
\dist_K(y,c_i) - \rho_i & = \dist_K(y,c_j) - \rho_j ,
\end{align*}
and thus
\[
\dist_K(y,c_i) - \dist_K(x,c_i) = \dist_K(y,c_j) - \dist_K(x,c_j) .
\]
But the left-hand side is $\dist_K(x,y)$ (because of the collinearity of $c_i$, $x$, and $y$), so we get
\[
\dist_K(x,y) = \dist_K(y,c_j) - \dist_K(x,c_j) .
\]
This however is impossible in general, because of the strict convexity of $K$, unless $x$, $y$ and $c_j$ 
are also collinear, that is, when $\lambda$ is the ray from $c_i$ through $c_j$. 
Regardless of which ray $\lambda$ is, since $f_i(c_i) < f_j(c_i)$ (by Lemma~\ref{lem:c0in}), 
$f_i(z) \le f_j(z)$ for all $z\in c_i x$ and $f_j (z) < f_i(z)$ for all $z\in\lambda\setminus c_ix$.

On the other hand, if $\lambda$ does pass through $c_i$ and $c_j$, then there is a unique point 
on the line $\bar\lambda$ supporting $\lambda$ that is equidistant from $K_i$ and $K_j$. 
	Indeed, by Lemma~\ref{lem:c0in}, $f_i(c_i) < f_j(c_i)$ and $f_j(c_j) < f_i(c_j)$,
so $B_{ij}$ intersects the segment $c_ic_j$ in at least one point, say $x$. As we 
move along $\bar\lambda$ from  $x$ toward $c_j$, the value of $f_j$ decreases but the value of 
$f_i$ increases until we reach $c_j$, and then, beyond $c_j$, both functions increase at the 
same rate, so $f_j(y)< f_i(y)$ for all points $y$ on $\bar\lambda$ after $x$ and, arguing symmetrically, 
$f_i(y)<f_j(y)$ for all points on $\bar\lambda$ before $x$. This completes the proof of the lemma.
\end{proof}

The following lemma is an easy corollary of Lemma~\ref{lem:mono}.
%---------------------------------------
\begin{lemma}
	\label{lem:star} 
	Let $\tau$ be a three-dimensional cell of $\A_{k}(\F)$
	lying on the graph of a function $f_i\in\F$, and let $\tau^\downarrow$ 
	denote the projection of $\tau$ onto $\reals^3$.  Then $\tau^\downarrow$ is 
	monotone with respect to $c_i$, in the sense that any ray emanating from 
	$c_i$ intersects $\tau^\downarrow$ in a connected (possibly empty or 
	unbounded) segment.
	%not containing $c_i$.
\end{lemma}
%---------------------------------------
\fi

We are now ready to describe the decomposition of the cells of $\A_k(\F)$ into elementary 
cells. (Our ultimate goal is to decompose the region below $\A_k(\F)$ into elementary cells,
but we start with this subtask.) For $1 \le i \le n$, let $\A_k^i = \A_k^i(\F)$ denote the subset of cells 
of $\A_k(\F)$ that lie on the graph of $f_i$, and let $\M_k^i$ denote their projections (which are the cells 
of $\M_k$ for which $K_i$ is the $k$-th nearest neighbor).
We describe the algorithm for decomposing the cells of $\M_k^i$ (or of $\A_k^i$).
Lemma~\ref{lem:mono} implies that, for any $j\ne i$, $B_{ij}$ can be viewed as the graph of a function 
$g_{ij}: \sphere^2 \rightarrow \reals_{\ge 0}$ in spherical coordinates with $c_i$ 
\ifsocgproc
as the origin.
	\else
as the origin.\footnote{%
  To make everything semi-algebraic, we should reparameterize the spherical coordinates
  $(\theta,\varphi)$ using some (standard) algebraic representation. Nevertheless, for 
  convenience, we stick to the spherical coordinate notation.}
\fi
That is, for a given $u \in \sphere^2$, let $\lambda (u)$ be the ray emanating from $c_i$ 
in direction $u$, and let $x_{ij}(u)$ be the intersection point of $\lambda(u)$ with the bisector $B_{ij}$
if such an intersection point exists; otherwise $x_{ij}(u)$ is undefined. We set 
$g_{ij}(u) = f_j(x_{ij}(u))$ if $x_{ij}(u)$ is defined and $+\infty$ otherwise. 
Let $G^{(i)} = \{ g_{ij} \mid 1 \le j \ne i \le n\}$. We define the \emph{level} of 
a point $x\in\reals^3$ in (the 3D arrangement, within $\sphere^2\times\reals$) 
$\A(G^{(i)})$ to be the number of functions 
in $G^{(i)}$ whose graphs intersect the (relatively open) segment $c_i x$.
The following lemma will be useful in analyzing the complexity of the decomposition of $\A_k^i$.

%\pankaj{Figure needs to be changed. One cannot have a function lying below $c_0$.}

%-----------------------------------------
\begin{lemma}
For a parameter $k<n$, let $\cell$ be a 3D cell of $\A_k^i$.
Then $\cell^\downarrow$ is a cell of level at most $k$ in $\A(G^{(i)})$.
\end{lemma}
%-----------------------------------------

\begin{figure}[htb]
  \begin{center}
	  \includegraphics[scale=0.8]{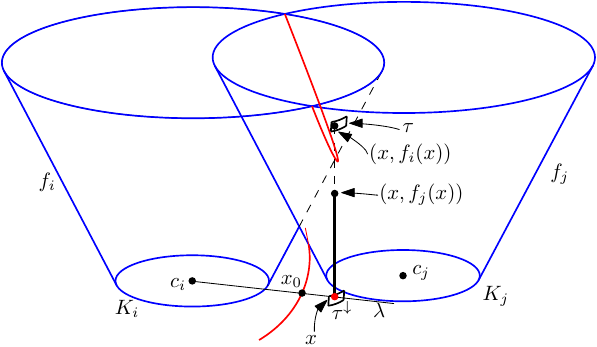} 
	  \caption{A bisector crossing $\lambda$. (The figure is drawn in the 3D $c$-space.)}
  \label{fig:before-after}
  \end{center}
\end{figure}

\begin{proof}
It is obvious from the definition of $G^{(i)}$ 
that $\cell^\downarrow$ is a cell in $\A(G^{(i)})$.
We now argue that its level in $\A(G^{(i)})$ is at most $k$.
Choose a point $x\in\cell^\downarrow$. Let $\lambda$ be the ray emanating from $c_i$
	and passing through $x$. Let $g_{ij}$ be a function whose graph 
intersects the segment $c_ix$, say, at a point $x_0 \in B_{ij}$. (See Fig.~\ref{fig:before-after}.)
By Lemma~\ref{lem:mono}, $f_i(y) < f_j(y)$ for all points $y$ on $\lambda$ preceding $x_0$,
and $f_j (y) < f_i (y)$ for all points on $\lambda$ beyond $x_0$. 
In particular, $f_i(x) > f_j(x)$. Since $\cell$ is a cell of $\A(\F)$,
$f_j$ does not intersect $\cell$, so $f_j$ lies below $\cell$ in 
$\reals^4$ (in the $x_4$-, or rather $\rho$-direction). Furthermore, 
$\cell \in \A_k(\F)$, so there are  exactly
$k$ functions of $\F$ whose graphs lie below $\cell$, implying that the level of $\cell^\downarrow$ 
in $\A(G^{(i)})$ is at most $k$, as claimed.
\end{proof}

\begin{figure}[htb]
  \begin{center}
	  \includegraphics[scale=0.8]{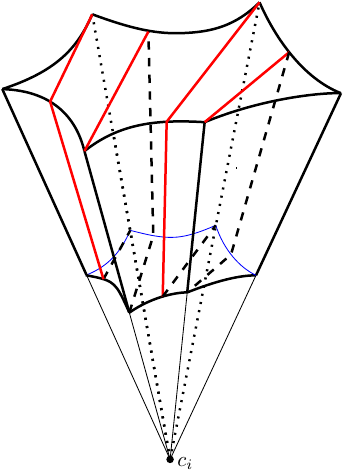} 
	  \caption{Vertical decomposition of a cell in the spherical coordinate system.}
  \label{fig:spherical-VD}
  \end{center}
\end{figure}

We decompose the cells $\cell$ of $\A_k^{(i)}(\F)$ by constructing ``vertical decompositions'' 
of the corresponding 3D projections $\cell^\downarrow$ in the spherical coordinate system
around $c_i$. As for standard vertical decompositions in 3D, we proceed in two stages. 
The first stage erects a wall from each edge of $\cell$ as follows.
Fix an edge $e$ of $\cell^\downarrow$. For a point $x\in e$, 
let $\lambda(x)$ be the ray emanating from $c_i$ and passing through $x$, and let 
$h(x)= \lambda(x)\cap\cell^\downarrow$. By Lemma~\ref{lem:mono}, $h(x)$ is a (connected) segment that
touches $e$ at one of its endpoints --- it is either the top endpoint,
for all points on $e$, or the bottom endpoint (with respect to distances from $c_i$). 
We erect the wall $\hat e = \bigcup_{x\in e} h(x)$ in $\cell^\downarrow$; 
parts of the wall may extend to infinity when $\cell^\downarrow$ is unbounded and $e$ is 
a bottom edge. We repeat this step for all edges of $\cell^\downarrow$.
These walls decompose $\cell^\downarrow$ into \emph{frustums} (truncated cones), each of 
which, denoted by $\Delta$, has a unique pair of \emph{front} and \emph{back} faces; the
other faces  of $\Delta$ lie on the created walls. (See Fig.~\ref{fig:spherical-VD}.) 
The projections of the front and back 
faces of $\Delta$ on $\sphere^2$ (with $c_i$ as the origin) are identical, and we denote 
this identical projection by $\Delta^\downarrow$. The complexity of $\Delta^\downarrow$ 
may be arbitrarily large (a typical situation in the first stage of vertical decompositions
in general; see, e.g., \cite{CEGS}). The subcell $\Delta$ itself is of the form 
\ifsocgproc
$\Delta = \{ (u,r) \mid u \in \Delta^\downarrow,\, r\in [g^-(u),g^+(u)] \}$,
\else
\[ 
\Delta = \{ (u,r) \mid u \in \Delta^\downarrow,\, r\in [g^-(u),g^+(u)] \} ,
\]
\fi
where $g^-, g^+$ are the functions of $G^{(i)}$ whose graphs contain the front and back faces of $\Delta$,
respectively. 

The second stage decomposes $\Delta$ into constant-complexity frustums, which we refer to as 
\emph{pseudo-cones}. We decompose $\Delta^\downarrow$ into spherical pseudo-trapezoids,
by drawing portions of meridians within $\Delta^\downarrow$ from each vertex and meridian-tangency 
point. We then lift each pseudo-trapezoid $\sigma^\downarrow$ to form the pseudo-cone 
$\sigma = \{ (u,r) \mid u \in \sigma^\downarrow, \, r \in [g^-(u),g^+(u)] \}$ in $\reals^3$.

Repeating this step for all frustums created in the first stage, we obtain
a decomposition of $\cell^\downarrow$ into pseudo-cones,
each being a semi-algebraic set of constant complexity and bounded by up to six facets
(compare with standard vertical decompositions in 3D~\cite{CEGS}). 
We refer to these pseudo-cones as \emph{elementary cells}. We note that each pseudo-cone 
$\phi$ is \emph{defined} by a subset $D_\phi$ of at most seven functions of $\F$ (the function $f_i$
and up to six additional functions that define the functions $g_{ij}$ that form the frustum), 
in the sense that $\phi$ appears in the vertical decomposition of a 3D cell of $\A(D_\phi)$.
By repeating this step for all cells of $\M_k^i$ and for all homothets $K_i$ of $\K$, 
we obtain a decomposition of $\M_k = \M_k(\K)$ into elementary cells. Finally, we 
lift each elementary cell of this decomposition vertically to $\A_k(\F)$ in a straightforward manner,
to obtain a corresponding decomposition of $\A_k(\F)$.

Kaplan~\etal~\cite[Lemma~6.2]{KMRSS} have shown that the complexity of the vertical decomposition of 
the $(\le k)$-level $\Ak{k}(\G)$ of a set of $\G$ of $m$ bivariate semi-algebraic functions of constant 
complexity is $O(k^2\psi(n/k)\lambda_s(k))$, where $\psi(r)$ is the maximum complexity of 
the lower envelope of a subset of $\G$ of size at most $r$, $\lambda_s(t)$ is the maximum 
length of Davenport-Schnizel sequences of order $s$ composed of $t$ symbols~\cite{SA}, 
and $s$ is a constant depending on the complexity of the functions in $\G$.
By applying the argument in \cite{KMRSS} to our spherical coordinate system and using
the worst-case bound $\psi(m) = O^*(m^2)$ on the complexity of the lower envelope
of $m$ constant-complexity bivariate functions~\cite{SA}, 
we conclude that the cells of $\M_k^i$ can be decomposed into $O^*(n^2k)$ elementary cells. 
Hence, $\A_k(\F)$, or $\M_k(\K)$, can be decomposed into $O^*(n^3k)$ elementary cells. 
By adapting the randomized incremental algorithm described in~\cite{KMRSS}, $\M_k^i$, for each 
$i$, can be computed in $O^*(n^2k)$ expected time. We thus obtain:

%---------------------------------------------
\begin{theorem} \label{thm:decomp}
Let $K$ be a compact, centrally-symmetric strictly convex semi-algebraic set in $\reals^3$ of constant 
complexity. Let $\K$ be a set of $n$ homothetic copies of $K$. The cells of $\M_k(\K)$, 
and the cells of the $k$-level in the arrangement of their $K$-distance 
functions, can be decomposed into $O^*(n^3k)$ elementary cells, in $O^*(n^3k)$ expected time.
\end{theorem}
%---------------------------------------------

%-----------------------------------
\section{Vertical Shallow Cuttings} \label{sec:shallow}

% \newtext{
A key notion that we need for constructing the dynamic NN data structures of Section~\ref{sec:ds}
is that of a vertical shallow cutting of $\A(\F)$, 
as studied in \cite{Chan:le,CT16,KMRSS} (for simpler scenarios). For a parameter $k$, a 
\emph{vertical $k$-shallow cutting} of $\F$ is a collection of pairwise openly disjoint 
semi-unbounded \emph{pseudo-prisms} (prisms for short), where each prism consists of all points 
that lie vertically below some pseudo-cone, of the decomposition (of Section~\ref{sec:vd}) 
that covers $\Ak{k} (\F)$,
and is thus a semi-algebraic set of constant complexity. 
Furthermore, the ceilings of these prisms collectively form an $x_4$-monotone surface that lies 
between $\A_k(\F)$ and $\A_{2k}(\F)$, and each prism is crossed by the graphs of at most 
$O(k)$ functions of $\F$. The \emph{conflict list} of a prism $\sigma$ is the set of 
functions whose graphs cross $\sigma$.

Following the ideas in~\cite{KMRSS}, for a parameter $t \in [1,n]$,
we construct a vertical $(n/t)$-shallow cutting as follows.
We set two parameters: $r=\beta t \log t$, where $\beta$ is a sufficiently large constant, independent
of $t$, and $q=\beta \log t = r/t$. We choose a random subset $\NN \subseteq \F$ consisting of $r$ functions,
construct $\A_q(\NN)$, and compute the decomposition of $\A_q(\NN)$ into a family $\Xi$ of pseudo-cones 
as described in Section~\ref{sec:vd}. For each pseudo-cone $\cell \in \Xi$, we associate a label 
$\varphi(\cell)$ which is the function of $\NN$ whose graph contains $\cell$. Finally, we erect a 
semi-unbounded prism $\cell^\uparrow$ from each pseudo-cone $\cell$ of $\Xi$, given by
\ifsocgproc
$\cell^\uparrow = \bigcup \{ \{c\} \times(-\infty, \rho] \mid (c,\rho) \in \cell\}$.
\else
\[
\cell^\uparrow = \bigcup \{ \{c\} \times(-\infty, \rho] \mid (c,\rho) \in \cell\}.
\] 
\fi
Let $\Xi^\uparrow$ be the resulting set of pseudo-prisms. By Theorem~\ref{thm:decomp}, with large probability,
\ifsocgproc
$|\Xi| = O^*(r^3q) =  O^*(t^3)$.
\else
\[
|\Xi| = O^*(r^3q) = O^*(t^3\log^4t) = O^*(t^3).
\]
\fi
We next show that $\Xi$ is a vertical $k$-shallow cutting of $\F$, where $k = n/t$, with high probability.

%------------------------------------------------
\ifsocgproc 
\subparagraph*{Range spaces and shallow $\eps$-nets.}
\else
\paragraph{Range spaces and shallow $\eps$-nets.}
\fi

Let $\Sigma=(\XX,\R)$ be a (finite) range space, where $\XX$ is a set of objects and $\R \subseteq 2^\XX$ a set of ranges.
Let $0 < \eps < 1$ be a given parameter. A subset $\NN \subseteq \XX$ is called 
a \emph{shallow $\eps$-net} of $\Sigma$ if it satisfies the following two properties
for every range $R \in \R$ and for any parameter $\zeta \ge 0$:

\medskip
\noindent{\bf (i)}
If $|R\cap \NN| \le \zeta \log \frac{1}{\eps}$ then $|R\cap \XX| \le \alpha(\zeta + 1)\eps|\XX|$, and

\medskip
\noindent{\bf (ii)}
If $|R\cap \XX| \le \zeta\eps|\XX|$ then $|R\cap \NN| \le \alpha(\zeta + 1) \log \frac{1}{\eps}$. 

\medskip
Here $\alpha$ is a suitable constant.
Note the difference between shallow and standard $\eps$-nets: Property (i) (with $\zeta = 0$) implies that
a shallow $\eps$-net is also a standard $\eps$-net (possibly with a recalibration of $\eps$). 
Property (ii) has no parallel in the case of standard $\eps$-nets --- there is no guarantee how 
a standard net interacts with small ranges (of the entire set $\XX$). The following result by 
Sharir and Shaul~\cite{ShSh} is a generalization of the result on standard $\eps$-nets~\cite{HW87}.

%-------------------------------------------------------------------
\begin{lemma}[Theorem 2.2 of Sharir and Shaul~\protect{\cite{ShSh}}]
	\label{lem:shallow-net}
Let $\Sigma=(\XX,\R)$ be a range space with VC-dimension $d$.
With a suitable choice of the constant of proportionality, a random sample $\NN \subseteq \XX$ of
${\displaystyle O\left(\frac{d}{\eps}\log\frac{1}{\eps} + \log\frac{1}{\delta}\right)}$
%\[
%O\left( \frac{1}{\eps} \left( \delta \log \frac{1}{\eps} + \log \frac{1}{\pi} \right) \right)
%\]
is a shallow $\eps$-net with probability at least $1 - \delta$.
\end{lemma}
%-------------------------------------------------------------------

\ifsocgproc 
\subparagraph*{$\Xi$ is a vertical shallow cutting.}
\else
\paragraph{$\Xi$ is a vertical shallow cutting.}
\fi
We apply the above result to the range space $\Sigma=(\F, \R)$, % where $\R$ consists of ranges,
so that each range in $\R$ is the subset of the surfaces of $\F$ that cross some region in $\reals^4$,
taken from some family $\Gamma$ of regions. Concretely, this includes the following families $\Gamma$:
% of several families $\Gamma$ of geometric objects in $\reals^4$. Each range in $\R$ is the 
% subset of those objects in such a family that intersect an object in $\Gamma$, i.e.,
% $\R = \{ \{ f_i \in \F \mid \gamma \cap f_i \ne\emptyset\} \mid \gamma \in \Gamma\}$.
(i) the set of pseudo-cones generated by the decomposition of a set of at most seven functions of $\F$ (recall that a pseudo-cone is defined by at most seven functions), 
(ii) the set of pseudo-prisms erected on these pseudo-cones, (iii) the set of edges in an arrangement 
of five functions of $\F$ (an edge in the arrangement of a subset of $\F$ 
is defined by at most five functions),
and (iv) the set of rays in the positive $x_4$-direction.
Using standard arguments, it can be shown that the VC-dimension of $\Sigma$ is finite. Therefore, by 
applying Lemma~\ref{lem:shallow-net} with $\eps=1/t$ and $\delta=1/3$ and choosing the constant of proportionality appropriately,
the random subset $\R$ is a $(1/t)$-shallow net of $\Sigma$ with probability at least $2/3$,
so we assume that $\NN$ is indeed a $(1/t)$-shallow net of $\Sigma$.
Recall that $q=\beta\log t$. Assuming $\beta\ge 2\alpha$, where $\alpha$ is the constant in the definition of shallow nets,
the converse of property (ii) of shallow nets, with $\zeta=1$, implies 
that the level of any point 
$p$ on $\A_q(\NN)$ with respect to $\F$ is  at least $n/t$.
Also, by property (i), the level of $p$  is at most $O(n/t)$. 
Finally, the shallow net property also implies that the size of the 
conflict-list of any prism in $\Xi^\uparrow$ is $O(n/t)$.
\iffullver
By testing each prism of $\Xi^\uparrow$ with every function of $\F$, 
the conflict list of all prisms can be computed in $O^*(nt^3)$ time.
\else
The conflict list of all prisms can be computed in $O^*(nt^3)$ time.
\fi
Hence, we obtain:

%---------------------------------------
\begin{theorem}
	\label{thm:cutting}
Let $K$ be a compact, centrally-symmetric strictly convex semi-algebraic set
in $\reals^3$ of constant complexity.  Let $\K$ be a set of $n$ homothetic 
copies of $K$. For any parameter $t \in [1,n-1]$, there exists a vertical 
$(n/t)$-shallow cutting $\Xi$ of the arrangement of the distance functions of $\K$ of size $O^*(t^3)$.
The cutting $\Xi$, along with the conflict list of each of its prisms, can be computed in $O^*(nt^3)$ expected time.
\end{theorem}
%---------------------------------------

%------------------------------
\section{Dynamic Data Structures for Proximity Queries} \label{sec:ds}

In this section we describe a dynamic data structure for answering \emph{intersection-detection}
queries amid $\K$. Namely, for a query homothet $K_0 \in \Kspace$, determine whether $K_0$ intersects
some homothet of $\K$, and, if the answer is yes, return such a homothet.
In addition to answering queries, the data structure can be updated efficiently
as a homothet $K_0\in\Kspace$ is inserted into $\K$ or deleted from $\K$. 
The data structure can be extended to answering \emph{intersection-reporting} queries, namely,
reporting all homothets of $\K$ that intersect $K_0$, at an additional cost of $O(k)$. 
\iffullver
Using the parametric-search technique
(see, e.g., \cite{AM93,AM95}), a nearest-neighbor query under the $K$-distance can be reduced 
to answering $O(\log n)$ intersection-detection queries, so the cost increases by only a logarithmic factor.
\fi

As already noted, the insertion of a new homothet into $\K$ can be handled using the standard 
dynamization technique by Bentley and Saxe~\cite{BS80} (see also~\cite{AM95}).
That is, assume that we have an intersection-detection data structure that can handle deletions,
which can be constructed in $P(n)$ time, so that a query can be answered in $Q(n)$ time and an object can be deleted in $D(n)$ time.
Then the amortized insertion time is $O((P(n)/n)\log n)$, the deletion time remains $D(n)$, 
and the overall query time is $O(Q(n))$, assuming that $Q(n) = \Omega(n^\eps)$.
\iffullver
See \cite{BS80} for details.
\fi

We first describe a data structure that answers a query in $O^*(1)$ time using
$O^*(n^3)$ space and handles a deletion in $O^*(n^2)$ amortized expected time, and then describe a linear-size data 
structure that answers a query in $O^*(n^{2/3})$ time and handles a deletion in $O(\log^2 n)$ amortized expected time. 
\iffullver
Interpolating between these two setups, using standard techniques~\cite{AAEKS}, 
for a storage parameter $s\in [n,n^3]$, 
we can construct a structure of size $O^*(s)$, in $O^*(s)$ expected time,
that answers a query in $O^*(n/s^{1/3})$;
the amortized expected update time is $O^*(s/n)$.
\fi

%------------------------------------------
\subsection{Fast query-time data structure}
\label{subsec:log}

We follow the same approach of Agarwal and Matou\v{s}ek for performing dynamic halfspace range 
reporting~\cite{AM95} (slightly improved but more involved approaches were proposed 
in~\cite{Chan:le,KMRSS}, but the one in~\cite{AM95} will do for our setting).  

\ifsocgproc 
\subparagraph*{Data structure.}
\else
\paragraph{Data structure.}
\fi
Let $\F$ be the collection of distance functions corresponding to the elements of $\K$.
Roughly speaking, we build a tree data structure on $\F$ using vertical shallow cuttings. 
That is, we start with $\F$, choose a parameter $t>1$, compute an $(n/t)$-vertical shallow 
cutting $\Xi$ of $\F$ of size $O^*(t^3)$, in $O^*(nt^3)$ time, 
create a child for each prism 
$\cell$ of $\Xi$, recursively construct the data structure for the conflict list 
$\F_\cell$ of every $\cell\in\Xi$, and attach it as the subtree of the child corresponding 
to $\cell$. The recursion proceeds until we reach cells $\cell$ with $|\F_\cell| = O(1)$,
in which case we just store $\F_\cell$ at $\cell$, as a list, say.

A query, with a copy $K_0 = (c_0,\rho_0)$ of $K$, is answered in a standard way, by
traversing the tree in a top-down manner. At each node (prism) $\cell$ that we visit,
we know that $(c_0,\rho_0)\in\cell$. We check whether $(c_0,\rho_0)$ lies in any
child prism of $\cell$, and, if so, recurse at that child. Otherwise, we locate the
child prism $\cell'$ for which $(c_0,\rho_0)$ lies above $\cell'$, and report any
element whose function belongs to the conflict list 
\iffullver
$\F_{\tau'}$ (or report all of them
in the reporting version).
\else
$\F_{\tau'}$.
\fi

However, this simple recursive approach runs into a technical difficulty
when we delete a function from $\F$.
When deleting a function from the root, the deletion has to be propagated to some of 
its children, to their children, and so on. On average, each function appears in 
the conflict lists of $O^*(t^2)$ prisms of $\Xi$. When such a ``good'' function is deleted, 
the deletion is propagated to only $O^*(t^2)$ descendants, leading to $O^*(n^2)$ overall deletion 
time, as can easily be verified. However, some of the functions (in fact, up to $n/t$ of them) may appear 
in the conflict lists of $\Omega(t^3)$ prisms of $\Xi$.
%(or at any rate $\omega(t^2)$ of them). 
If these ``bad'' functions are the first $n/t$ functions to be deleted and then the data structure 
is reconstructed, the overall cost of the deletion operations will be too high. 
As in~\cite{AM95,KMRSS}, we circumvent this difficulty by maintaining a partition of $\F$ into a few 
subsets and constructing a cutting for each subset that is good for all functions in that subset 
(i.e., each function in the subset appears in only $O^*(t^2)$ conflict lists). 

We now describe the data structure in detail. We fix a sufficiently large constant $t>0$.
The data structure is periodically reconstructed after performing some deletions. We use 
$m$ to denote the size of $\F$ when the data structure was previously reconstructed and $n$ 
to denote its current size. We reconstruct the data structure after deleting $m/2t$ functions 
from $\F$. Thus $n \ge m(1-1/2t)$. Let $P(r)$ denote the maximum time spent in constructing the 
data structure for a set of $r$ functions. We pay for the reconstruction cost by charging 
$2tP(m-m/2t)/m \le 2tP(n)/n$ time to each of the $m/2t$ delete operations that occurred 
before the reconstruction (the inequality holds because $P$ is assumed to be superlinear). 
The data structure for $\F$ is a (recursively defined) tree $\Psi(\F)$. We describe how a 
subtree $\Psi(\G)$, for a subset $\G$ of $\nu$ functions, is constructed.

If $\nu \le n_0$, for some sufficiently large constant $n_0$, then $\Psi(\G)$ is a single
leaf, and we simply store $\G$ at that leaf. So assume that $\nu > n_0$. The root of $\Psi(\G)$ 
stores the following data:

\medskip
\begin{itemize}
	\item A partition of $\G$ into subsets $\G_1, \ldots, \G_u$, $u \le \ceil{\log_2 t}$, 
%whose construction is described below.
as described below.
	\item For each $1 \le i \le u$, a vertical $(\nu/t)$-shallow cutting 
$\Xi_i$ for $\G_i$. 
For each $\Delta \in \Xi_i$, we store its conflict list $\G_{i,\Delta}$.
	\item For each $1 \le i \le u$ and $\Delta\in\Xi_i$, a pointer to the tree $\Psi(\G_{i,\Delta})$.
	\item For each function $f\in \G_i$, the set $L_f\subseteq\Xi_i$ of prisms $\Delta$
		crossed by the graph of $f$.
		%, (i.e., $f \in \G_{i,\Delta}$).
	\item A counter $\chi$ that keeps track of how many functions of $\G$ can be deleted before
		we reconstruct $\Psi(\G)$.
\end{itemize}

\medskip
After the construction of $\Psi(\G)$, before any deletions occur, the following properties hold:

\medskip
\begin{enumerate}[(P1)]
	\item The counter $\chi$ is set to $\nu/2t$.
	\item For each $i$, the $x_4$-monotone surface formed by the pseudo-cones of $\Xi_i$ 
(i.e., the ceilings of prisms of $\Xi_i$) lie above $\A_{\nu/t} (\G_i)$.
	\item For every $i$ and $\Delta\in \Xi_i$, $|\G_{i,\Delta}| \le cn/t$, 
for some suitable absolute constant $c$.
	\item Each function $f\in \G_i$ appears in the conflict lists of at most
    $\kappa := 2C_1t^{2+\delta}$ 
	prisms of $\Xi_i$, where $C_1$ and $\delta>0$ are the constants appearing in
		(\ref{eq:avg-cross}) below.
\end{enumerate}

We now describe the construction of the partition $\G_1, \ldots, \G_u$, which proceeds
iteratively. Suppose we have constructed $\G_1, \ldots, \G_{i-1}$. Set 
$\G_i^* = \G \setminus \bigcup_{j < i} \G_i$ and put $\nu_i = |\G_i^*|$. 
If $\nu_i \le \nu/t$, we stop. Otherwise, set $t_i = t\nu_i/\nu$ and $k=\nu_i/t_i=\nu/t$.

We construct a vertical $k$-shallow cutting $\Xi_i$ of $\G_i^*$ of size $O(t_i^{3+\delta})$,
for an arbitrarily small constant $\delta>0$. 
For each $\Delta\in \Xi_i$, $|\G_{i,\Delta}^*| \le c\nu_i/t_i = c\nu/t$. We note that
\begin{equation}
\label{eq:avg-cross}
\sum_{\Delta\in\Xi_i} |\G_{i,\Delta}^*| \le O(t_i^{3+\delta}) \frac{c\nu}{t} 
\le C_1 \left ( \frac{t\nu_i}{\nu}\right)^{3+\delta} \frac{\nu}{t} 
\le C_1 \nu_i t^{2+\delta} .
\end{equation}
A function $f \in \G_i^*$ is called \emph{good} if it appears in the conflict list of at 
most $\kappa$ prisms of $\Xi_i$ and \emph{bad} otherwise. 
Let $\G_i$ be the set of good functions of $\G_i^*$; set $\G_{i+1}^* = \G_i^*\setminus \G_i$.
By~(\ref{eq:avg-cross}), $|\G_{i+1}^*| \le \nu_i/2$. Hence, $u \le \ceil{\log_2 t}$. 
It is easily seen that properties (P1)--(P4) hold after the construction.
This completes the description of the data structure. The total size and preprocessing time are $O^*(n^3)$.

\ifsocgproc 
\subparagraph*{Query procedure.}
\else
\paragraph{Query procedure.}
\fi
Let $K(c_0,\rho_0)\in\Kspace$ be a query homothet.
Recall that $K(c_0,\rho_0)$ intersects a homothet of $\K$ if the point $(c_0,\rho_0)$ lies above the 
lower envelope of the corresponding set $\F$ of functions. The query is performed as follows. 
We search $\Psi(\F)$ in a top-down manner with $(c_0,\rho_0)$. Suppose we are at the root of a subtree
$\Psi(\G)$. If $\Psi(\G)$ is a leaf, we check all functions of $\G$ in a brute-force manner and 
return an answer in $O(n_0)=O(1)$ time. Otherwise, for each $i \le u$, we check whether $(c_0,\rho_0)$ 
lies above the lower envelope of $\G_i$. If the answer is yes for any $i$, 
we return yes and also return an intersecting homothet (see below).  Otherwise, we return no.
%(Returning an intersecting homothetic copy, in case an intersection is detected, is easy, and we omit the details.)

For a fixed $i\le u$, we proceed as follows.
We search with $c_0$, in a brute force manner, in $O^*(t^3) = O(1)$ time, 
to determine the pseudo-cone $\cell^\downarrow$ in  the projection of $\Xi_i$
that contains $c_0$. Let $f_i$ be the  function associated with $\cell$, i.e., 
the ceiling of $\Delta$ lies in the graph of $f_i$.
We test whether $\rho_0 \ge f_i(c_0)$. If so, $K_i$ and $K(c_0,\rho_0)$ intersect, so we
return yes and report $K_i$ as an intersection witness and terminate the query. If not, $(c_0,\rho_0)$
lies in the semi-unbounded prism $\Delta$,
and we continue the search recursively in the data structure $\Psi(\G_{i,\Delta})$.

Let $Q(\nu)$ be the maximum time spent by the query procedure in the subtree $\Psi(\G)$ 
storing at most $\nu$ functions. Then we obtain the following recurrence:
%$Q(\nu) \le \ceil{\log_2 t} Q(c\nu/t) + O(t^{3+\delta})$.
\[ Q(\nu) \le \ceil{\log_2 t} Q(c\nu/t) + O(t^{3+\delta}) .\]
Its solution is easily seen to be $Q(\nu) = O(n^\eps)$, for any constant $\eps>0$, provided that 
\ifsocgproc
$n_0$ and $t$ are chosen sufficiently large.
\else 
	$n_0$ and $t$ are chosen sufficiently large.\socgcomm{\footnote{%
			As shown in~\cite{AM95}, the query time can be improved to $O(\log n)$ by choosing $t=n^\delta$. 
  	The same improvement can be obtained here, but, for simplicity of exposition, we stay with 
  	the above weaker result.}}
  \fi

\ifsocgproc 
\subparagraph*{Deletion procedure.}
\else
\paragraph{Deletion procedure.}
\fi
Let $f\in\F$ be a function that we wish to delete from $\F$.
Again, we visit $\Psi(\F)$ in a top-down manner. Suppose we are at the root $v$ of a subtree $\Psi(\G)$. 
If $v$ is a leaf, then we simply delete $f$ from $\G$ and stop. If $v$ is an internal node, we first 
decrement the counter $\chi$ stored at $v$. If $\chi$ becomes $0$, we reconstruct $\Psi(\G)$ from
scratch (for the current $\G$). Otherwise, we find the index $i$ such that $f\in \G_i$. 
For each $\Delta\in L_f$,  we delete $f$ from $\G_{i,\Delta}$ and then recursively delete 
$f$ from $\Psi(\G_{i,\Delta})$. By construction, $|L_f|\le \kappa$. The deletion procedure 
maintains the properties (P3) and (P4), and (P2) is replaced with a slightly weaker property:

\smallskip
\begin{description}
	\item[(P2')] For every $i$, the $x_4$-monotone surface formed by the pseudo-cones of $\Xi_i$ 
        (i.e., the ceilings of prisms of $\Xi_i$) lie above or on $\A_{\nu/2t} (\G_i)$.
\end{description}

The correctness of the query procedure follows from property (P2'). Following the analysis 
in~\cite{AM95}, the amortized deletion time, including the time spent in reconstructing the data 
structure, is $O^*(n^2)$. 
Putting everything together, and combining this with the technique of Bentley and Saxe~\cite{BS80}
for insertions, we obtain:

%------------------------------------
\begin{lemma}
	\label{lem:log-qtime}
%Let $\K$ be a set of $n$ homothets of $K$.  
$\K$ can be preprocessed,
in $O^*(n^3)$ expected time, into a data structure of size $O^*(n^3)$, so that 
an intersection-detection query, including the cost of reporting an intersecting member of $\K$ 
if one exists, can be performed in $O^*(1)$ time. 
A homothet can be inserted into $\K$ or deleted from $\K$ in $O^*(n^2)$ amortized expected time.
\end{lemma}
%------------------------------------

%-----------------------------------------------------
\subsection{Linear-size data structure}
\label{subsec:linear}

Next, we present a linear-size dynamic data structure that supports intersection-detection queries 
in $O^*(n^{2/3})$ time each, and can handle updates (insertions and deletions) in $O(\log^2 n)$ 
amortized expected time per update.

For a homothet $K(c_0,\rho_0)$, let $f_{c_0,\rho_0}^+ = \{(c,\rho) \in \reals \mid \rho \ge f_{c_0,\rho_0}(c) \}$ 
be the region lying above the graph of the distance function $f_{c_0,\rho_0}$, which corresponds to the set of 
homothets that intersect $K({c_0,\rho_0})$. Put $\K^* = \{(c,\rho) \mid K(c,\rho) \in\K\}\subset\reals^4$.
For a query homothet $K(c_0,\rho_0)$, we wish to determine whether any point of $\K^*$ 
lies in $f_{c_0,\rho_0}^+$. We construct a linear-size partition tree for answering these queries, 
following the same approach as in~\cite{AM95,Mat92,ShSh}. We need a few definitions. 
Let $\Fspace$ be the set of distance functions corresponding to the homothets in $\Kspace$.
Let $P \subset \Kspace$ be a set of $n$ points (representing homothets of $K$).
For a parameter $k\ge 1$, we call a semi-algebraic set 
$\gamma \subset \Kspace$, which is semi-unbounded in the $(-x_4)$-direction,
\emph{$k$-shallow} if $|P \cap \gamma| \le k$. 
A major ingredient of the approach in~\cite{Mat92,ShSh} is to construct a 
so-called \emph{test set} composed of a small number of semi-algebraic sets, which
represent well (in the sense made precise below) all distance functions of $\Fspace$ 
that are $(n/r)$-shallow, for a given parameter $r>1$, with respect to $P$. 
Formally, a finite collection $\Phi$ of semi-algebraic sets of constant complexity (not necessarily 
a subset of $\Fspace$) is called a \emph{test set} for $(n/r)$-shallow ranges in $\Fspace$ with 
respect to $P$ if it satisfies the following properties: 

\medskip
  \begin{description}
\item[(i)]
  Every set in $\Phi$ is $(n/r)$-shallow with respect to $P$. 
\item[(ii)]
  The complement of the union of any $m$ sets of $\Phi$ can be decomposed 
  into at most $\zeta(m)$ ``elementary cells'' (semi-algebraic sets of constant complexity) for any $m\ge 1$, 
  where $\zeta(m)$ is a suitable monotone increasing superlinear function of $m$. 
\item[(iii)]
  Any $(n/r)$-shallow set $\sigma\in \Fspace$ can be covered by the union 
  of at most $\delta$ ranges of $\Phi$, where $\delta$ is a constant (independent of $r$).
\end{description}

\medskip
Sharir and Shaul~\cite{ShSh} showed that if there exists a test set of size $r^{O(1)}$ for $(n/r)$-shallow 
ranges of $\Fspace$, then one can construct a linear-size partition tree on $P$ so that an
$\Fspace$-emptiness query can be answered in $O^*(n/\zeta^{-1}(n))$ time, for the corresponding function $\zeta$. 

\ifsocgproc 
\subparagraph*{Test-set construction.}
\else
\paragraph{Test-set construction.}
\fi
We describe an algorithm for constructing a test set 
of size $O(r^4)$ with $\zeta(m)=O^*(m^3)$ and $\delta=1$.
Let $\F$ be the set of distance functions of $\K$ as above. We take a 
random subset $R\subseteq \F$ of $s = cr\log r$ functions, for some sufficiently large constant $r$, 
where $c>0$ is another sufficiently large absolute constant, independent of, and much smaller than $r$, 
and construct the (standard) 4D
vertical decomposition $\VeD{\A}(R)$ of the arrangement $\A(R)$, of size $O^*(r^4)$, as 
described in~\cite{AES,Koltun-04a,SA}.
By a standard random-sampling argument~\cite{HW87}, with probability at least $1/2$, 
each cell of $\VeD{\A}(R)$ is crossed by at most $n/r$ functions of $\F$, 
provided that $c$ is chosen sufficiently large.
If this is not the case, we discard $R$ and choose another random subset, until we find, in
expected $O(1)$ trials, one with
the desired property. 
We clip $\VeD{\A}(R)$ within the halfspace $x_4\ge 0$ (to restrict it to within $\Kspace$).
We choose a subset $\Xi$ of the cells of $\VeD{\A}(R)$, consisting of those cells that have at most 
$n/r$ functions of $\F$ passing \emph{fully below} them. By construction, these cells cover 
$\Ak{n/r}(\F)$, and are contained in $\Ak{2n/r}(\F)$.

Let $\cell$ be a cell of $\Xi$.
Set $\Phi_\cell = \bigcup_{(c_0,\rho_0)\in\cell} f_{c_0,\rho_0}^+$.
Since $\cell$ and $f_{c_0,\rho_0}$ are semi-algebraic sets of constant complexity, 
$\Phi_\cell$ is also a semi-algebraic set of constant complexity. 
For a homothet $K_i\in \K$, the point $K_i^* \in \Kspace$ lies in $\Phi_\cell$ if and only 
if there is a point $(c_0,\rho_0)\in\cell$ such that $K_i^* \in f_{c_0,\rho_0}^+$. This happens
when $(c_0,\rho_0)\in f_i^+$, i.e., the graph of 
the function $f_i$ crosses $\cell$ or lies below $\cell$.
By construction, there are at most $n/r + n/r = 2n/r$ such functions. 
Consequently, $\Phi_\cell$ is $(2n/r)$-shallow with respect to the points of $\K^*$ with
$\delta = 1$.

Set $\Phi = \{\Phi_\cell \mid \cell \in \Xi\}$.
$\Phi$ is a family of $O^*(r^4)$ constant-complexity semi-algebraic 
sets
in $\Kspace$, each of which is $(2n/r)$-shallow with respect to $\K^*$. This is the desired test set.
Each set $\Phi_\cell$ is unbounded in the positive $x_4$-direction, therefore the 
complement of the union of a subset of $m$ sets is the region lying below their 
lower envelope. By a result in~\cite{AES:vd}, the complement of their union (i.e., the 
region below their lower envelope) can be decomposed into $O^*(m^3)$ pseudo-prisms. Hence, we obtain:
%----------------------------------
\begin{lemma} \label{lem:btest}
Let $P\subset \Kspace$ be a set of $n$ points and $r\ge 1$ a parameter. A test set $\Phi$ 
	of size $O(r^4)$ for the functions in $\Fspace$ that are $(n/r)$-shallow 
	with respect to $P$ can be computed in $O^*(r^4)$ expected time, so that
	$\zeta(m)=O^*(m^3)$ and $\delta=1$.
\end{lemma}
%----------------------------------

Using Lemma~\ref{lem:btest} and adapting the approach in~\cite{ShSh}, we can preprocess $\K$ into
a linear-size data structure that supports intersection-detection queries in $O^*(n^{2/3})$ time. 
By (re)constructing portions of the partition tree periodically (as in the previous data structure), 
deletions can be handled efficiently. Omitting the straightforward details, we obtain:

%------------------------------------
\begin{lemma}
	\label{lem:linear-qtime}
Let $\K$ be a set of $n$ homothets of $K$.  $\K$ can be preprocessed,
in $O(n\log n)$ expected time, into a data structure of size $O(n)$, so that an intersection-detection query 
can be answered in $O^*(n^{2/3})$ time. A homothet of $K$ 
can be inserted into or deleted from the data structure in $O(\log^2 n)$ amortized time.
\end{lemma}
%------------------------------------

By combining this data structure with Lemma~\ref{lem:log-qtime} in a 
standard manner~\cite{AAEKS}, we can 
obtain the following space/query-time trade-off:
%------------------------------------
\begin{theorem}
	\label{thm:tradeoff}
	Let $\K$ be a set of $n$ homothets of $K$, and let $s\in[n,n^3]$ be a storage parameter.  
        $\K$ can be preprocessed, in $O^*(s)$ expected time, into a data structure of size $O^*(s)$, so 
	that an intersection-detection query can be answered 
	in $O^*(n/s^{1/3})$ time. A homothet can be inserted into 
	or deleted from the data structure in $O^*(s/n)$ amortized time.
\end{theorem}
%------------------------------------

\iffullver
\begin{remark} 
	\label{rem:ds}
	(i) Using standard techniques~\cite{Mat92,AM95}, the above data structure can be adapted in a 
        straightforward manner to report all $k$ homothets intersected by 
	a query homothet in an additional $O(k)$ time.

	(ii) The above theorem, with $s = n^{3/2}$, implies that a sequence of $n$ insertion/deletion/query 
        operations on the data structure can be performed in $O^*(n^{3/2})$ expected time.
\end{remark}
\fi

Finally, by combining this data structure with the parametric-search technique~\cite{AM93}, 
we establish Theorem~\ref{thm:NN} (for answering nearest-neighbor queries in $\K$ under the $K$-distance).

%------------------------------------
% \begin{theorem}
	% \label{thm:NN}
	% Let $\K$ be a set of $n$ homothets of $K$, and let $s\in[n,n^3]$ be a storage parameter.  $\K$ can be preprocessed,
% in $O^*(s)$ expected time, into a data structure of size $O(s)$, so that for any point $x\in\reals^3$, its nearest neighbor in $\K$ under the $K$-distance can be computed
	% in $O^*(n/s^{1/3})$ time. A homothet can be inserted into $\K$ or deleted from $\K$ 
	% in $O(s/n)$ amortized time.
% \end{theorem}
%------------------------------------

\iffullver
\section{Breadth-First Search, Reverse Shortest Paths, and Related Problems} 
\label{sec:bfs-rsp}

In this section we present applications of our techniques to various basic graph algorithms
on intersection or proximity graphs of a set of homothets in $\reals^3$.

In this section we present applications of our techniques to the reverse shortest path problem.

\ifsocgproc 
\subparagraph*{Breadth-first search.}
\else
\paragraph{Breadth-first search.}
\fi
We follow the standard high-level approach. We run the BFS layer by layer, where the first layer
contains just the start object $K(c_s,\rho_s)$. Consider the step where we have constructed layer 
$L_i$ and wish to construct the next layer $L_{i+1}$. Let $\U_i$ denote the set of all members of 
$\K$ that the BFS has not yet reached (up to this step). We maintain $\U_i$ in the dynamic data 
structure of Section~\ref{sec:ds}, which supports deletions and queries (insertions are irrelevant 
here), where a query specifies an object $K(c_0,\rho_0)$ 
(in $L_i$), and wants to determine whether $K(c_0,\rho_0)$ intersects some object $K(c,\rho)$ of 
$\U_i$, and, if so, to return one such object (or, alternatively, report all intersecting homothets
in $\U_i$). The output object, or objects, if they exist, are added to $L_{i+1}$,
and are immediately removed from the structure (and from $\U_i$), to make sure that they are not reported 
again (by other elements of $L_i$, or by elements of future layers).

%We fix the storage parameter $s$ (to a value that will be specified shortly). A query then takes $O^*(n/s^{1/3})$ time, and a deletion takes (amortized) $O^*(s/n)$ time. Since 
The BFS procedure executes a total of $O(n)$ queries and deletions, therefore by Remark~\ref{rem:ds} 
(following Theorem~\ref{thm:tradeoff}, with $s=n^{3/2}$),
the overall cost of the the algorithm is $O^*(n^{3/2})$, as stated in
Theorem~\ref{thm:bfs}~(a).

\ifsocgproc
\subparagraph*{Depth-first search.}
\else
\paragraph{Depth-first search.}
\fi
The same machinery can be used to perform depth-first search on the graph $\G(\K)$.
At each step we are at some homothet $K(c_0,\rho_0)\in\K$, and want to find an intersecting neighbor 
that has not yet been visited. As above, we maintain the set $\U$ of all members of $\K$ that 
the DFS has not yet reached, using the dynamic data structure of Section~\ref{sec:ds}, with the
same choice of $s=n^{3/2}$. Using this structure, we find a neighbor of $K(c_0,\rho_0)$
(if one exists), delete it from $\U$, and continue the search at that neighbor (or back up
to the parent node, or start a new DFS tree). As above, since we perform
$O(n)$ searches and deletions, the overall cost is $O^*(n^{3/2})$ time. 
This, together with the preceding analysis, establish Theorem~\ref{thm:bfs}~(a).

%\ifsocgproc
%\subparagraph*{Reverse shortest paths.}
%\else
%\paragraph{Reverse shortest paths.}
%\fi
\paragraph{Reverse shortest paths.}
We next use Theorem~\ref{thm:bfs} to solve the corresponding reverse shortest path (RSP) problem,
defined as follows.
Let $\K$ be as above, and let $r>0$ be a parameter. We consider the setup where, for each 
$K(c,\rho)\in\K$, we expand its size by adding $r$ to the size of each element of $\K$, 
i.e., replacing each $K(c,\rho)$ by $K(c,\rho+r)$.
Define $\G_r(\K)$ to
be the intersection graph of the expanded homothets in $\K$. We are then given two elements
$K(c_s,\rho_s)$, $K(c_t,\rho_t)$ of $\K$ and an integer $1\le k < n$, and seek the smallest value
$r^*$ for which $\G_{r^*}(\K)$ contains a path from $K(c_s,\rho_s+r^*)$ to $K(c_t,\rho_t+r^*)$,
of at most $k$ edges. This is the \emph{reverse shortest path} (RSP) problem for that graph.

The decision problem associated with this RSP problem is, given $r$, to decide whether $\G_r(\K)$
contains such a path. If the answer is positive (resp., negative) then $r^*\le r$ (resp.,
$r^* > r$). Solving this decision problem is trivially performed using BFS on $\G_r(\K)$, 
starting from $K(c_s,\rho_s+r)$, so, by Theorem~\ref{thm:bfs}, the resulting decision procedure
runs in $O^*(n^{3/2})$ time.

We now solve the RSP problem using the shrink-and-bifurcate technique (see \cite{BFKKS,CH,KKSS}),
adapted to the setup at hand. The \emph{critical values} of $r$ are values at which some comparison involving
$r$ that the BFS algorithm performs changes its outcome. These critical values are of several kinds. 
First, these are 
values of $r$ at which two expanded homothets, $K(c,\rho+r)$, $K(c',\rho'+r)$, become externally 
tangent to one another (they are disjoint for smaller values of $r$ and intersect from this $r$ 
and above). This happens when
\begin{equation} \label{eq:crit} 
\dist_K(c,c') = (\rho+r) + (\rho'+r) = \rho + \rho' + 2r, \;\;\text{or} \;\; 
r = \frac12 \left( \dist_K(c,c') - \rho - \rho' \right) .
\end{equation}

Other kinds of critical values arise in comparisons that occur during the construction of, and 
search in, the data structure of Section~\ref{sec:ds}. There are many kinds of such comparisons,
each of which either involves the manipulation of $O(1)$ functions $f_{c,\rho}$ (that is, 
$f_{c,\rho+r}$), or comparisons between points and functions, which determine whether the
point lies above, on, or below the function. The latter kind of comparisons gives rise to
critical values of the first kind (given in (\ref{eq:crit})). 

The shrinking part of the technique receives a parameter $L$, and aims to construct an interval $I$
that contains $r^*$ and at most $L$ other critical values. This procedure is based on a variant 
of ``distance selection'', which in turn is based on a procedure for counting the number of 
critical values that are smaller than or equal to $r_0$, for some query parameter $r_0$. 
This procedure is implemented using off-line semi-algebraic range searching. 
Specifically, for critical values in (\ref{eq:crit}), each homothet
$K(c,\rho)$ is mapped to the point $(c,\rho)\in\reals^4$, as before, and also to the range
\[
\sigma(c,\rho) = \{ (c',\rho') \in\reals^4 \mid f_{c,\rho}(c') = \dist_K(c,c') - \rho \le \rho' + 2r_0 \} .
\]
The goal is then to count the number of point-range containments. Each range is a semi-algebraic region
of constant complexity, and both points and ranges have four degrees of freedom. The general technique
for offline semi-algebraic range searching, as presented, e.g., in \cite[Appendix]{AAEKS}, runs in
this case in time $O^*(n^{8/5})$, which implies that, using the improved procedure 
of Chan and Huang~\cite{CH}, interval shrinking (restricted to this type of critical values)
can be performed in $O^*(n^{8/5}/L^{4/5})$ expected time.

To handle critical values of other kinds, we observe that our additive expansion rule
subtracts the same value $r$ from each of these functions, so comparisons involving any 
set of $O(1)$ functions can be determined independently of $r$. Put differently, we can
think of the process as constructing the data structure without $r$ (i.e., with $r=0$),
in a manner that requires no simulation, and then query the structure with $K(c_0,\rho_0+2r)$.
Hence the only $r$-dependent comparisons occur during a query.

Unfortunately, here there is a significant difference between the fast-query version of
the structure and the linear-space version. In the fast-query version, the query is with 
the point $(c_0,\rho_0+2r)$, and each $r$-dependent comparison tests the position of this point
(above, on, below) with respect to a function in $\F$. In other words, all the corresponding
critical values are of the kind in (\ref{eq:crit}). In contrast, in the linear-storage version, 
the query is with a surface, dual to the point $(c_0,\rho_0+2r)$, and certain steps in the 
execution of a query seek cells of the current decomposition that are intersected
by the surface or lie fully below it, and these operations, which do depend on $r$, generate
critical values that are not of the type (\ref{eq:crit}).

In principle, this is not a real problem. As just briefly reviewed, the shrinking stage 
(see \cite{BFKKS,CH,KKSS}), relies on a ``distance-selection'' procedure, which is
implemented as an off-line semi-algebraic range-counting procedure, where the ranges
are defined by some semi-algebraic predicate of constant complexity. For critical
values of the former kind, the predicate is given by (\ref{eq:crit}). For critical
values of the other kinds, we just have to use different predicates. Here too the
queries have four degrees of freedom (defined by the parameters $(c_0,\rho_0+2r)$), 
but the data objects (e.g., a vertex of a cell in the decomposition) may have many more 
degrees of freedom, and this will slow down the procedure.

There are two ways out of this difficulty. One approach would be to bound the maximum 
number of degrees of freedom of the data objects, which is some small constant, and then run the relevant range 
searching procedure, using, e.g., the analysis in \cite[Appendix]{AAEKS}, with the
inferior time bound that it would yield. The other approach, which is the one that we follow,
is to give up in the simulation the use of the linear-space version of the structure,
and just use the fast-query version.

Recall that this version, applied to a set of $n$ homothetic copies, uses a total of $O^*(n^3)$
storage and preprocessing, answers a query in $O^*(1)$ time, and performs a deletion in
$O^*(n^2)$ (amortized) time. To make it more efficient, we partition $\K$ into $O(n/t)$ subsets, 
each consisting of at most $t$ homothets, for a suitable parameter $t$ that will be fixed shortly,
and construct a different structure for each subset. 
A query has now to be performed on each of these structures, for a total of $O^*(n/t)$ time,
while a deletion has to be performed only within its own subset, for a cost of $O^*(t^2)$
amortized expected time. Choosing $t=n^{1/3}$ makes the two bounds asymptotically the same, 
for a total cost (for $O(n)$ queries and deletions) of $O^*(n^{5/3})$.

The standard analysis of the bifurcation stage~\cite{BFKKS} yields the bound $O^*(L^{1/2}D(n))$ on its
running time, where $D(n)$ is the cost of the decision procedure. Here however, we need to
modify the analysis because we have two versions of the decision procedure, the one that we 
simulate, and the one that runs, without simulation, on some concrete values $r_0$ of $r$.
The respective running time bounds of these versions are $D_1(n) = O^*(n^{5/3})$ and
$D_0(n) = O^*(n^{3/2})$.

Here is a sketch of the modified analysis (see \cite{BFKKS,KKSS} for full details). 
The standard analysis uses two parameters $X$, $Y$, where $X$ is a threshold on the 
number of unresolved comparisons accumulated in a single phase of the bifurcation, 
and $Y$ is a threshold
on the total amount of simulation steps performed in a phase. The cost of a phase is
$O^*(XY + D_0(n))$, where the second term results from the binary search for $r^*$ at the 
end of the phase. We take $XY = D_0(n)$. The number of phases is 
${\displaystyle O\left( \frac{L}{X} + \frac{D_1(n)}{Y} \right)}$, which we optimize by making 
these two terms equal, i.e., by choosing ${\displaystyle Y = \frac{D_1(n)X}{L}}$. Solving this 
system of equations, we get ${\displaystyle X = L^{1/2}\left( \frac{D_0(n)}{D_1(n)}\right)^{1/2}}$, 
and the overall cost is then
\[
O^*\left( XY\cdot \frac{D_1(n)}{Y} \right) = O^*(D_1(n)X) = O^*(L^{1/2}D_0(n)^{1/2}D_1(n)^{1/2}) .
\]
In our case, $D_1(n) = O^*(n^{5//3})$ and $D_0(n) = O^*(n^{3/2})$, so the overall bound is
\[
O^*\left( \frac{n^{8/5}}{L^{4/5}} + L^{1/2}n^{3/4}n^{5/6} \right) =
O^*\left( \frac{n^{8/5}}{L^{4/5}} + L^{1/2}n^{19/12} \right) .
\]
To balance the two terms, we choose $L^{13/10} = n^{1/60}$, or $L = n^{1/78}$, and the bound then becomes
$O^*\left(n^{62/39} \right)$. That is, we have obtained the bound asserted in Theorem~\ref{thm:rsp}.
%---------------------------------------------
% \begin{theorem} \label{thm:rsp}
% The reverse shortest path in the intersection graph $G(\K)$ can be performed in $O^*(n^{224/143})$ time.
% \end{theorem}
%---------------------------------------------

\medskip
%\noindent{\bf Remark.}
\begin{remark}
When $K$ is the Euclidean ball, the algorithm becomes simpler, and its performance can be
improved to $O^*(n^{56/39})$, as recently shown in \cite{KSS}. The reason is that the intersection
detection subprocedure can be transformed to a dynamic halfspace range reporting in $\reals^5$,
using a simple lifting transform. The latter task can be solved using the machinery of Agarwal and
Matou\v{s}ek~\cite{AM95}, which leads to the above improved bound. See \cite{KSS} for full details.
\end{remark}

\ifsocgproc
\subparagraph*{Minimum spanning forests and Dijkstra's shortest path algorithm.}
\else
\paragraph{Minimum spanning forests and Dijkstra's shortest path algorithm.}
\fi

Let $P$ be a set of $n$ points in $\reals^3$, and let $K$ be as above. As we recall,
for a given $r_0>0$, 
the $r_0$-proximity graph $\G_{r_0}(P)$ (with respect to the metric $\|\cdot\|_K$) 
is the graph whose vertices are the points of $P$, and whose
edges are all the pairs $(p,q)$ such that $\dist_K(p,q) \le r_0$. 
% This is in fact the intersection graph $\G(\K)$, where $\K = \{ K(q,r_0/2) \mid q\in P\}$. 
We assign to each edge $(p,q)\in G_{r_0}(P)$ the weight $\dist_K(p,q)$. We consider the problems 
of computing a minimum spanning forest of $\G_{r_0}(P)$ (note that $\G_{r_0}(P)$ might be disconnected), 
and of computing the shortest-path tree in $\G_{r_0}(P)$ from some start point $K_s$.

For the minimum spanning forest problem, 
we run Prim's algorithm which, at each step, maintains the partition of $P$
into the subsets $V$, $U$, of visited and unvisited points, respectively. The algorithm computes
the bichromatic closest pair (BCP) $(p,q)$ in $V\times U$. If $\dist_K(p,q) \le r_0$, 
we add the edge $(p,q)$ to the MST and move $q$ from $U$ to $V$. Otherwise, we
conclude that the vertices of $U$ form a connected component of $\G_{r_0} (P)$, we discard $U$, 
and continue running Prim's algorithm on $V$ (this case will not arise when $\G_{r_0}(P)$ 
is connected). By Corollary~\ref{cor:bcp}, the BCP of $U$ and $V$ can be maintained in $O^*(n^{1/2})$ 
amortized expected time under insertion/deletion of points, so the overall expected run time 
of the algorithm is $O^*(n^{3/2})$, as claimed in Theorem~\ref{thm:bfs}~(b).

The implementation of Dijkstra's shortest-path algorithm is similar.
At each step of the algorithm, we have the sets $V$ and $U$ of nodes that the algorithm has already visited, and
those that it did not, respectively. Each point $p\in V$ has a weight $\omega(p)$ associated with it, equal to
the shortest distance from the start point $s$ to $p$. We now compute the BCP in $V\times U$, where the
points $p\in V$ are additively weighted by the corresponding distances $\omega(p)$, and the points
of $U$ are unweighted. At each step, we wish to compute the weighted BCP of $V \times U$,
namely $(p^*,q^*) = \arg\min_{(p,q)\in V\times U} \left( \dist_K(p,q)+\omega(p) \right)$. 
If $\dist_K(p^*,q^*) \le r_0$, 
we add $q^*$ as a child of $p^*$, and move $q^*$ from $U$ to $V$, with weight
$\omega(q^*) = \omega(p^*) + \dist_K(p^*,q^*)$. Otherwise, we conclude that all remaining 
points of $U$ are farther than $r_0$ from $p^*$, so none of them is adjacent to $p^*$.
We delete $p^*$ from $V$ and continue.

%Again, using Eppstein's technique, we need dynamic data structures for finding the nearest neighbor in $U$ of a point in $V$, and for finding the nearest neighbor in $V$ of a point in $U$.  The first structure is the one of Section~\ref{sec:ds} for the unweighted point set $U$, and the second structure is on the weighted set $V$. To handle the second case, we replace each point
We implement the computation of $(p^*,q^*)$  using Corollary~\ref{cor:bcp}, as follows.
We choose a sufficiently large parameter $A$, larger than $\max_p \omega(p)$, which can easily be 
computed in advance in $O(n)$ time. Then we replace each point $p\in V$ by the homothetic copy 
$K_p = p + (A-\omega(p))K$. Set $\K_V = \{ K_p \mid p \in V\}$. It is easily seen that if 
$(K_{p^*}, q^*)$ is the BCP of $\K_V \times U$, then $(p^*,q^*)$ is the weighted BCP of 
$V$ and $U$. Hence, each step of the
algorithm can be implemented in $O^*(n^{1/2})$ amortized expected time. 

At each step of the algorithm, we either move a point from $U$ to $V$ or delete a point from $V$.
Since each point of $P$ is moved from $U$ or deleted from $V$ at most once, the overall 
expected time is $O^*(n^{3/2})$, as claimed in Theorem~\ref{thm:bfs}~(b).
\fi

%\micha{Do we want to mumble about the Voronoi diagram under the $K$-distance? Other related stuff?}

%------------------------------------------

\end{document}